\documentclass{article}
\usepackage{hyperref, url}
\usepackage{amscd,amsmath,amsthm,amssymb}
\usepackage{enumerate,varioref}
\usepackage{epsfig}
\usepackage{graphicx}
\usepackage{subfig}
\usepackage{mathtools}
\usepackage{tikz}
\usepackage{algorithmicx}
\usepackage{algorithm}
\usepackage{algpseudocode}
\usepackage{xcolor}
\newtheorem{thm}{Theorem}

\theoremstyle{definition}
\newtheorem{defn}[thm]{Definition}
\theoremstyle{remark}

\newtheorem*{rem}{Remark}

\newcommand{\Z}{\mathbb{Z}}

\newcommand{\upket}{|\uparrow\rangle}
\newcommand{\upbra}{\langle\uparrow|}
\newcommand{\downket}{|\downarrow\rangle}
\newcommand{\downbra}{\langle\downarrow|}

\title{Consistent Rotation Maps Induce a Unitary Shift Operator in Discrete Time Quantum Walks}
\author{Clark Alexander\\ email: \href{mailto:gcalexander1981@gmail.com}{the author}}

\begin{document}
	
	\maketitle
	
	\begin{abstract}
		In this work we explain the necessity for consistently labeled rotation maps for efficiently computing coined discrete time quantum walks on regular graphs.
	\end{abstract}

\section{Introduction}
In Discrete Time Quantum Walks, hereafter DTQW, on regular graphs, one can compute more efficiently by the use of a ``quantum coin."  In this case a DTQW becomes the evolution of the vector $|\psi_0\rangle$ via unitary operator.  In the case of \cite{S} this can be achieved through reflections and shifts. In the case of a regular graph, one can reduce this further to $\hat{U} = \hat{S}\hat{C}$ where $\hat{S}$ is the ``shifting operator" which respects the graph structure, and $\hat{C}$ is the ``coin operator" which controls the propagation of waves emanating from each vertex.  In this work, we only wish to consider the construction of a unitary operator $\hat{S}$.

\section{The Main Theorem}
\subsection{The Construction of $\hat{S}$}
We will begin with the simplest possible example, that of a square.
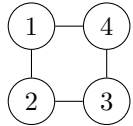
\[
\begin{tikzpicture}
\node[circle, draw](1) at (0,1) {1};
\node[circle, draw](2) at (0,0) {2};
\node[circle, draw](3) at (1,0) {3};
\node[circle, draw](4) at (1,1) {4};
\foreach \from/\to in {1/2,2/3,3/4,4/1}
\draw (\from) -- (\to);

\end{tikzpicture}
\]

This is a 2-regular graph which means we have two coin states. From a historical perspective we call these states $\upket$ and $\downket$

Our shift operator is defined by
\begin{eqnarray}
\hat{S} & = & \upket \upbra \otimes |2\rangle\langle 1| + \upket \upbra \otimes |3\rangle\langle 2| \\
\nonumber &+& \upket \upbra \otimes |4\rangle\langle 3| +\upket \upbra \otimes |1\rangle\langle 4| \\
\nonumber &+& \downket \downbra \otimes |4\rangle\langle 1| + \downket \downbra \otimes |1\rangle\langle 2| \\
\nonumber &+& \downket \downbra \otimes |2\rangle\langle 3| +\downket \downbra \otimes |3\rangle\langle 4|
\end{eqnarray}

We can easily verify that this is unitary $\hat{S} \hat{S}^* = I$.
As we have it written the $\upket$ state sends us counterclockwise around the square and $\downket$ sends us clockwise.  If we had instead used a ``greedy labeling" where $\upket$ sends us to the smaller vertex label and $\downket$ to the larger vertex label then we would have this operator

\begin{eqnarray}
\hat{S}_{\text{incorrect}} & = & \upket \upbra \otimes |2\rangle\langle 1| + \upket \upbra \otimes |1\rangle\langle 2| \\
\nonumber &+& \upket \upbra \otimes |2\rangle\langle 3| +\upket \upbra \otimes |1\rangle\langle 4| \\
\nonumber &+& \downket \downbra \otimes |4\rangle\langle 1| + \downket \downbra \otimes |3\rangle\langle 2| \\
\nonumber &+& \downket \downbra \otimes |4\rangle\langle 3| +\downket \downbra \otimes |3\rangle\langle 4|
\end{eqnarray}

We can check
\[
\hat{S}_{\text{incorrect}} \hat{S}^*_{\text{incorrect}} = \begin{bmatrix}
2 & 0 & 0 & 0 & 0 & 0 & 0 & 0\\ 
0 & 2 & 0 & 0 & 0 & 0 & 0 & 0\\ 
0 & 0 & 0 & 0 & 0 & 0 & 0 & 0\\ 
0 & 0 & 0 & 0 & 0 & 0 & 0 & 0\\ 
0 & 0 & 0 & 0 & 0 & 0 & 0 & 0\\ 
0 & 0 & 0 & 0 & 0 & 0 & 0 & 0\\ 
0 & 0 & 0 & 0 & 0 & 0 & 2 & 0\\ 
0 & 0 & 0 & 0 & 0 & 0 & 0 & 2
\end{bmatrix}
\]
which is certainly not the Identity. So what happened?  The short answer is that we did not label the edges consistently, but rather, greedily or naively.  

Reasoning by example, for a general graph $G = (V,E)$ with regularity $d$, we have the shift operator defined by

\begin{equation}
\hat{S} = \sum_{j=1}^d \sum_{v\in V} |c_j\rangle \langle c_j| \otimes |w\rangle \langle v|
\end{equation} 

In this case the vertices $w$ are adjacent to the vertices $v$.  From a computational point of view we can more easily compute this by considering the \emph{rotation map} in matrix form.

\begin{defn}
	Given a regular graph $G = (V,E)$ with degree of regularity $d$, the rotation map on $G$ is a function
	\begin{equation}
	Rot_G : [|V|] \times [d] \rightarrow [|V|]\times [d]
	\end{equation}
	where $Rot_G(v,i) = (w,j)$ with the $i^{th}$ edge leaving from vertex $v$ enters vertex $w$, and the $j^{th}$ edge leaving vertex $w$ enters vertex $v$.
\end{defn}

and looking at \cite{A} we see the reduction to matrix form.

\begin{defn}
	Given a regular graph $G$ with degree of regularity $d$ we redefine the rotation map as a matrix $Rot^{(G)} \in \Z^{|V|\times d}$
	by
\begin{equation}
	Rot^{(G)}_{v,i} = w \label{rotation matrix}
\end{equation}
	where $Rot_G(v,i) = (w,j)$ as above. Since $Rot_G$ is an involution the information $j$ is redundant.
\end{defn}

With this in mind, we define the shifting operator as follows
\begin{equation}
\hat{S} =  \sum_{j=1}^d \sum_{v\in V} |c_j\rangle \langle c_j| \otimes |Rot^{(G)}_{v,j}\rangle \langle v| \label{shift-rotate}
\end{equation}

Reading this in slightly mathematically reduced language we have:\\

``When we `flip' our quantum coin and it lands on side $j$, our wave at vertex $v$ propagates to the adjacent vertex via the edge labeled $j$."

\subsection{The Rotation Map Should be Consistent}

One may easily read a rotation map from the adjacency matrix as follows:

	\begin{algorithm}
	\caption{Rotation map from the adjacency matrix; Julia/Octave style}
	\label{Adjacency to rotation}
	\begin{algorithmic}
		\State for $k$ in $1:|V|$
		\State $\hspace{0.5cm}$ counter = 1
		\State for $l$ in $1:|V|$
		\State $\hspace{0.5cm}$if $A[k,l] == 1$
		\State $\hspace{1cm} Rot^{(G)}[v,\text{counter}] = l$
		\State $\hspace{1cm}$ counter $+= 1$
		\State endif
		\State endfor
		\State endfor
	\end{algorithmic}	
\end{algorithm}

Pairing this algorithm with \ref{shift-rotate} we arrive at the greedy shifting operator, which we have seen is not unitary.  This will provide some sort of dynamics on the graph, but it will not be quantum mechanical since the probabilities will sum to $d^t >> 1$ rather than 1.  In order to keep things quantum mechanical we must consider a consistent edge labeling of $G$, which brings us to a consistent rotation map.

\begin{defn}
	Let $G=(V,E)$ be a graph with regularity $d$. We say $G$ is consistently edge labeled if each vertex has $d$ distinct incoming labels.  A Rotation map which encodes a consistent graph labeling is said to be a \emph{consistent rotation map}.  Moreover, this is an easily checkable criterion in the matrix formulation.  That is, each column of $Rot^{(G)}$ has $d$ distinct elements.
\end{defn}

\begin{rem}
	In Python we can check that each column has $d$ distinct elements by the call
	\[
	\text{numpy.sort(RotG, axis=0)}
	\]
	In Julia the call is
	\[
	\text{sort(RotG, dims = 1)}
	\]
\end{rem}

Now we come to the main theorem
\begin{thm}
	Let $G=(V,E)$ be a graph with regularity $d$ and rotation map $Rot^{(G)}$.  The shifting operator $\hat{S}$ defined by
	\[
	\hat{S} =  \sum_{j=1}^d \sum_{v\in V} |c_j\rangle \langle c_j| \otimes |Rot^{(G)}_{v,j}\rangle \langle v|
	\]
	is unitary if and only if $Rot^{(G)}$ is consistent.
	
\end{thm}

\begin{proof}
First let's consider the oeprator $\hat{S}^*$

\begin{equation}
\hat{S} =  \sum_{j=1}^d \sum_{v\in V} |c_j\rangle \langle c_j| \otimes |v\rangle \langle Rot^{(G)}_{v,j}|
\end{equation}

Now we consider the product $\hat{S}\hat{S}^*$.  There are a few obvious reductions as we have defined the bases of our Hilbert spaces $\mathcal{H}_\mathcal{P}$ and $\mathcal{H}_C$ to be orthonormal. Thus we have
\[
\langle c_j | c_k \rangle  = \delta_{jk}, \text{ and } \langle j|k \rangle = \delta_{jk}
\]

Let's restrict our attention to a single basis element of $\mathcal{H}_C$ in the product $\hat{S}\hat{S}^*$

\[
|c_j\rangle\langle c_j| \otimes \left( \sum |Rot^{(G)}\rangle \langle v| \right)  \left(\sum_{w\in V} |w\rangle\langle Rot^{(G)}| \right)
\]

In the operator $\hat{S}$ we have a exactly one copy of $\langle v|$ for each basis state in the coin space.  If we don't have $d$ distinct labels in the Rotation map, then the diagonal of our matrix will contain at least one zero, and we'll have a nontrivial kernel.  

More explicitly
\[
\hat{S}\hat{S}^* = \sum |Rot^{(G)}_{v,j}\rangle \langle Rot^{(G)}_{v,j}|
\]
This produces a projection for each iteam in the $j^{th}$ column of the rotation map in matrix form.  The identity is simply the sum of $d$ unique projections.

If we have a consistent rotation map then $\hat{S}\hat{S}^*$ will be the identity.  If we do not have a consistent rotation map then $\hat{S}\hat{S^*}$ will not be the identity.  Thus $Rot^{(G)}$ is a consistent rotation map if,and only if $\hat{S}$ is unitary.

\end{proof}

\section{Technical Difficulties}

It is an open question as to whether there is a constructive polynomial time algorithm to construct a consistent rotation map on a random regular graph.  We know how to construct consistent rotation maps on several families of regular graphs, but we do not know how to construct this in general.  The author has built several heuristic solvers and has had some success with smaller graphs with lower degrees of regularity, but even a moderately sized graph for example 80 vertices, 12 regular evade a fully consistent rotation map.

\end{document}